\documentclass[letterpaper, 10 pt, conference]{ieeeconf}  

\IEEEoverridecommandlockouts                              
\usepackage{amsfonts}
\usepackage{fancyhdr}

\usepackage{times}
\usepackage{amsmath}
\usepackage{changepage}
\usepackage{amssymb}
\usepackage{graphicx}
\usepackage{floatflt}
\usepackage{verbatim}
\usepackage{amsrefs}
\usepackage{tabularx}
\usepackage[dvipsnames]{xcolor}
\usepackage{colortbl}
\usepackage{bm}
\usepackage[colorlinks]{hyperref}
\usepackage{url}
\usepackage{comment}

\usepackage[framemethod=tikz]{mdframed}
\usetikzlibrary{arrows,intersections,chains}
\usetikzlibrary{positioning}

\newtheorem{theorem}{Theorem}

\newtheorem{corollary}[theorem]{Corollary}

\newtheorem{definition}[theorem]{Definition}
\newtheorem{example}[theorem]{Example}

\newtheorem{proposition}[theorem]{Proposition}
\newtheorem{remark}[theorem]{Remark}

\newcommand{\E}{\mathbb{E}}

\newcommand{\argmin}{\text{argmin}}

\newcommand{\mc}{\mathcal}

\title{\LARGE \bf
Risk-Averse Equilibrium Analysis and Computation
}

\author{Ilia Shilov, H\'el\`ene Le Cadre, and Ana Busic
\thanks{I. Shilov and A. Busic are with INRIA-ENS Paris, France
        {\tt\small \{ilia.shilov, ana.busic\}@inria.fr}}%
\thanks{H. Le Cadre is with VITO/EnergyVille, Thor Park, Belgium
        {\tt\small helene.lecadre@energyville.be}}%
}

\begin{document}

\maketitle
\thispagestyle{plain}
\pagestyle{plain}

\begin{abstract}
We consider two market designs for a network of prosumers, trading energy: (i) a centralized design which acts as a benchmark, and (ii) a peer-to-peer market design. High renewable energy penetration requires that the energy market design properly handles uncertainty. To that purpose, we consider risk neutral models for market designs (i), (ii), and their risk-averse interpretations in which prosumers are endowed with coherent risk measures reflecting heterogeneity in their risk attitudes. We characterize analytically risk-neutral and risk-averse equilibrium in terms of existence and uniqueness, relying on Generalized Nash Equilibrium and Variational Equilibrium as solution concepts. To hedge their risk towards uncertainty and complete the market, prosumers can trade financial contracts. We provide closed form characterisations of the risk-adjusted probabilities under different market regimes and a distributed algorithm for risk trading mechanism relying on the Generalized potential game structure of the problem. The impact of risk heterogeneity and financial contracts on the prosumers' expected costs are analysed numerically in a three node network and the IEEE 14-bus network.
\end{abstract}


\section{Introduction}

The increasing amount of Distributed Energy Resources (DERs), which have recently been integrated in power systems, and the more proactive role of consumers have transformed the classical centralized power system operation by introducing more uncertainty and decentralization in the decisions. Following this trend, electricity markets are starting to restructure, from a centralized market design in which all the operations were managed by a central market operator, to more decentralized designs involving local energy communities which can trade energy by the intermediate of the central market operator or in a peer-to-peer (p2p) setting \cite{lecadre, moret, moret_pinson, tushar, wang}. P2p energy trading is a novel paradigm, where prosumers providing their own energy from DERs such as solar panels, wind turbines, combined heat and power (CHP), gas boilers, storage technologies, demand response
mechanisms \cite{busic}, etc., exchange energy with one another. The development of p2p energy trading has significant potential that can benefit end users, both in terms of revenue generation and in terms of reducing the cost of electricity, as well as reducing grid dependence \cite{tushar}.

\subsection{Related work}
Consideration of a novel market organization allowing p2p energy trading requires a careful quantitative comparison with the existing centralized design which is currently implemented in most of the existing pool-based  markets. The economic dispatch of a local energy community under different structures of communications is analysed using consensus based approaches in \cite{moret, moret_pinson}. Game theoretic approaches integrating the prosumers' strategic behaviors are considered in \cite{tushar, wang} through auctions and to quantify the efficiency loss in \cite{lecadre}. To analyse the market in presence of shared coupling constraints, we employ Generalized Nash Equilibrium (GNE) as solution concept \cite{harker, kulkarni, yin}, and a refinement of it, called Variational Equilibria (VE), assuming the shadow variables associated with the shared coupling constraints are aligned among the agents \cite{kulkarni, rosen}. 

Inclusion of different risk perceptions of the agents in the energy trading model calls for risk-augmented market formulations in which agents are endowed with risk measures in their utility functions. Large literature exists on the impact of risk on agent decisions in electricity markets \cite{abada, hoschle, philpott, gerard, moret, oum, ralph}. Additionally, heterogeneous description of uncertainties on agent assets make the market incomplete \cite{moret}. In order to hedge their risk towards uncertainties and complete the market, electricity market participants can choose among several financial products and derivatives as reviewed in \cite{oum}. 

We consider risk-adjusted market with heterogeneous risk attitudes of the agents, who are willing to maximize their risk-adjusted utility. As argued in Philpott et al. in \cite{philpott}, if agents have different views on future outcomes, then the social welfare optimization might not correspond to the competitive equilibrium solution. This raises the question on how one might complete the market with suitable financial instruments to enable a welfare maximizing competitive equilibrium. We include financial contracts as instruments to reduce the effect of heterogeneous risk attitudes on the outcome of the risk adjusted market and analyse the regimes of the market with respect to the exogenous and endogenous contracts' prices.


\subsection{Contributions}
Heterogeneous risk-adjusted market is comprehensively considered in \cite{moret}. Moret et al. address the definition of fairness and the impact of risk in a one settlement two-stage p2p market. They provide a model of risk hedging mechanism via financial contracts and analyse their impact on fairness and agents' payments. In this paper, we extend the given contract mechanism, providing a distributed implementation of the mechanism. Furthermore, relying on game-theoretic solution concepts such as GNE, VE, we characterise the p2p market outcomes, deriving non trivial results on existence and uniqueness of market equilibrium. 

The difference between centralized and p2p markets is extensively addressed in \cite{le cadre}, where authors capture the imperfections in information of the agents' caused by p2p organization by employing GNE and VE. In this paper, we extend the quantification of efficiency loss considering a risk-adjusted p2p market. 

In \cite{gerard}, Gerard et al. employ coherent risk measures. They analyse risk-adjusted markets and evaluate the impact of risk-hedging contracts on the market efficiency. They address the question of uniqueness of risk-averse equilibria and give insight on some equivalences between social planner problems and equilibrium problems. Using the same risk-averse setting, we extend their model by considering the strategic behaviors of the prosumers. Furthermore, we provide a new direction for the distributed implementation of financial contracts and risk-averse equilibrium computation, allowed by the Generalized Potential Game structure of our problem.

The organization of the rest of this paper is as follow: in Sections \ref{statement} and \ref{agents}, we introduce our model. In Section \ref{two market designs}, we analyse two market designs: risk-neutral in Subsection \ref{risk neutral}, where we study the difference between nodal prices in centralized and p2p risk neutral cases, and risk-adjusted in Subsection \ref{risk averse} in which we compute the agents' strategies in a risk-adjusted electricity market. In Section \ref{equilibrium problems} we analyse equilibria in an incomplete market and provide results on the coincidence of the centralized design outcome and VE assuming constant congestion costs. We prove uniqueness of VE in the risk-neutral setting under the general assumption of linearity of congestion costs. In Section \ref{completeness of the market}, we quantify the impact of inclusion of financial risk-hedging contracts on the market outcome and derive closed form expressions linking the risk trading market outcome to exogenous insurance price level. Two test cases are presented in Section \ref{numerical}: a three node toy network and the standard IEEE 14-bus network. We evaluate numerically the impact of prosumers' heterogeneous risk attitudes and the impact of risk-hedging contracts. 

\subsection{Notations}

Bold symbol $\boldsymbol{x}$ denotes a vector and capital italic symbol $\mathcal{X}$ denotes a set. Exclusion of a set is denoted by
$\mathcal{X}\setminus \mathcal{A} := \{x | x \in \mathcal{X}, x \notin  \mathcal{A}\}$. $\Omega$ represents a finite set of scenarios, and $X(\omega)$ denotes the realization of the random variable $X$ according to the scenario $\omega\in\Omega$. $\mathbb{E}$ denotes the expectation with respect to the discrete distribution $(p_{\omega})_{\omega \in \Omega}$, $R$ refers to a risk measure. 

\section{Statement of Problem} \label{statement}

We consider a single-settlement market for p2p energy trading made of a set $\mc{N}$ of $N$ prosumers -- each one of them being located in a node of the distribution grid. In a single-settlement market, each prosumer $n$ chooses independently her bilateral trades $\boldsymbol{q}_n$, at a cost $\tilde{C}(\boldsymbol{q}_n)$, energy generation $G_n$, incurring an uncertain production cost $C(G_n)$ and demand $D_n$, to maximize her utility function $\Pi_n$.

We assume that there are a finite number of scenarios $\omega \in \Omega$. The bilateral trades negotiated within prosumers, the flexibility activated by the prosumers and the demand all depend on the scenario $\omega$ defining the realized \textit{ target demand} $D^{\ast}_n(\omega)$ and RES-based generations $\Delta G_n(\omega)$ at each node $n$.

\section{Agents} \label{agents}

\subsection{Prosumers}

Let $\Gamma_n$ be the set of neighbors of $n$, with the structure of a communication network (local energy
community). It does not necessarily reflect the grid constraints. As usual, we assume that $n \in \Gamma_n$, for
all $n \in \mc{N}$. In particular, $\Gamma_0 := \mc{N} \setminus \{0\}$.

In each node, we introduce $\mathcal{D}_n := \{D_n \in \mathbb{R}_{+}| \underline{D}_n \leq D_n \leq \overline{D}_n\}$ as agent $n$'s demand set, with $\underline{D}_n$ and $\overline{D}_n$ being the lower and upper-bounds on demand capacity. 
$\Delta G_n$ as the random RES-based generation at node $n$. 

In parallel to the demand and RES-based generation-sides, we define the self-generation-side by letting $G_n := \{ G_n \in \mathbb{R}_+ | \underline{G}_n \leq
G_n \leq \overline{G}_n \}$ be agent n’s generation set, where $\underline{G}_n$
and $\overline{G}_n$ are the lower and upper-bounds on generation capacity.

The decision variables of each prosumer $n$ are her demand $D_n$, energy generation $G_n$, and the quantity exchanged between $n$ and $m$ in the direction from $m$ to $n$, $q_{mn}$, for all $m \in \Gamma_n \setminus \{n \}$. 
If $q_{mn} \geq 0$, then $n$ buys $q_{mn}$ from $m$, otherwise ($q_{mn} < 0$) $n$ sells $-q_{mn}$ to $m$. We impose an inequality on the trading reciprocity:
\begin{equation}\label{reciprocity}
q_{mn} + q_{nm} \leq 0, \qquad \forall m \in \Gamma_n,
\end{equation}
which means that, in the case where $q_{mn} > 0$, the quantity that $n$ buys from $m$ can not be larger than the quantity $q_{nm}$ that $m$ is willing to offer to $n$. Let $\kappa_{nm} \in [0, +\infty )$ be the equivalent interconnection capacity between node $n$ and node $m$, such that $\kappa_{nm} = \kappa_{mn}$ and $G_n \leq \kappa_{nm}$ and 
\begin{equation}\label{eq:capacity}
    q_{nm} \leq \kappa_{nm}, \qquad \forall m \in \Gamma_n
\end{equation}

\subsubsection{Local Supply and Demand Balancing}
Local supply and demand equilibrium leads to the following equality in each node $n$ in $\mathcal{N}$:
\begin{eqnarray}\label{eq:supply_demand_eq}
D_n(\omega) &=& G_n(\omega) + \Delta G_n(\omega) + \sum_{m \in \Gamma_n}  q_{mn}(\omega), \nonumber \\
&=& G_n(\omega) + \Delta G_n(\omega) + Q_n(\omega), \forall \omega \in \Omega. 
\end{eqnarray}

\subsubsection{Cost and Usage Benefit Functions}

Flexibility activation (production) cost in node $n$ is modeled as a quadratic function of local activated flexibility, using three positive parameters $a_n$, $b_n$ and $d_n$:
\begin{equation}\label{eq:prod_cost}
C_n\big(G_n(\omega)\big) = \frac{1}{2} a_n G_n(\omega)^2 + b_n G_n(\omega) +d_n, \forall \omega \in \Omega.
\end{equation}

We make the standard assumption  that self-generation occurs at zero marginal cost.

The usage benefit perceived by agent $n$ is modeled as a strictly concave function of node $n$ demand, using two positive parameters $\tilde{a}_n, \tilde{b}_n$ and a target demand $D^{\ast}_n (\omega)$, defined exogenously for agent $n$:
\begin{equation}
    U_n\big(D_n(\omega)\big) = -\tilde{a}_n(D_n(\omega) - D^{\ast}_n (\omega))^2 + \tilde{b}_n, \forall \omega \in \Omega.
\end{equation}
We consider that prosumers have preferences on the possible trades with their neighbors. The preferences are modeled with  (product) differentiation prices: each agent $n$ has a price $c_{nm}(\bm{q})$ modeled as a function of trading desicions of agents on the edge $(n,m)$ to buy energy from an agent $m$ in her neighborhood $\Gamma_n$. The total trading cost function of agent $n$ is denoted by:
\begin{equation}\label{eq:cost_preference}
\tilde{C}_n(\bm{q}_n) = \sum_{m \in \Gamma_n,m\neq n} c_{nm}(\bm{q}) q_{mn},
\end{equation}
where $c_{nm}(\bm{q}) =  a_{mn} (q_{mn} + q_{nm})+ b_{mn}, \,\, b_{mn} > 0,  a_{mn}=a_{nm} \in \mathbb{R}$ -- edge-specific cost described as a function of the trading decisions on the edge $(n,m)$. Coefficient $a_{mn}$ may be interpreted as a physical property of the link $(n,m)$ equal for both agents $m$ and $n$, for example how the traded flows affect the congestion on the edge, and term $b_{mn}$ captures agent-specific preferences. 

\begin{example}\label{ex1}
Consider the case where $a_{mn}= 0 $, that is -- $c_{nm}(\boldsymbol{q}) = b_{mn} >0$ is a constant, which provides interesting economics interpretations. Parameters $b_{mn}$ can model taxes to encourage/refrain the development of certain technologies (micro-CHPs, storage, solar panels) in some nodes. They can also capture agents' preferences to pay regarding certain characteristics of trades (RES-based generation, location of the prosumer, transport distance, size of the prosumer, etc.). If $q_{mn} >0$ (i.e., $n$ buys $q_{mn}$ from $m$) then $n$ has to pay the cost $b_{mn}q_{mn} >0$. Thus, the higher $b_{mn}$ is, the less interesting it is for $n$ to buy energy from $m$ but the more interesting it is for $n$ to sell energy to $m$. 
\end{example}

\subsubsection{Prosumer $n$'s cost function}

We introduce the revenue generated by the selling of the aggregated surplus to the grid operator as $p_0(\omega) \sum_{n \in \mathcal{N}} Q_n$, where $p_0(\omega)$ is the price of the import as defined on the balancing (real-time market) and $Q_n(\omega) := \sum_{m \in \Gamma_n} q_{mn}(\omega)$ is defined as the difference between the sum of imports and the sum of exports in node $n$. $Q_n$ will be called the net import at node $n$. It is proved in \cite{lecadre} that at the equilibrium, the total sum of the net imports at all nodes should be negative or null, i.e., $\sum_{n\in \mc{N}} Q_n \leq 0$. Prosumer $n$'s cost function is defined as the difference between the usage benefit resulting from the consumption of $D_n$ energy unit and the sum of the flexibility activation, trading costs and $p_0 Q_n$, which represents an allocation of the profit resulting from the surplus selling on the balancing  to prosumer $n$ by the aggregator proportionally to the prosumer's contribution to the surplus. Formally, it takes the form:
\begin{multline} \label{eq:profit}
    \Pi_n(\omega) = C_n\big(G_n(\omega)\big) + \tilde{C}_n\big(\bm{q}_n(\omega)\big) +\\
    + p_0(\omega)Q_n(\omega) - U_n\big(D_n(\omega)\big) , \forall \omega \in \Omega,
\end{multline}

where $\bm{q}_n(\omega) = \big(q_{mn}(\omega)\big)_{m \in \Gamma_n, m \neq n}$.
    
\subsection{Local Market Operator (MO)}

We introduce the social cost as the sum of the cost functions of all the agents in $\mathcal{N}$:
\begin{equation} \label{eq:social_welfare}
SC(\omega) = \sum_{n \in \mathcal{N}} \Pi_n(\omega).
\end{equation}

\section{Two Market Designs}\label{two market designs}

In the following sections, we assume existence of solutions of the considered optimization problems to be given. We consider notions of Generalized Nash Equilibria and Variational Equilibria: both of them exist under mild conditions \cite{kulkarni}, \cite{yin}. 

We introduce a risk measure $R$, that is a functional that associates to a random utility function its deterministic equivalent, i.e., the deterministic utility deemed as equivalent to the random loss. We assume a coherent risk measure \cite{artzner}, meaning that it satisfies four natural properies: monotonicity (if $\boldsymbol{X} \geq \boldsymbol{Y}$ then $R[\boldsymbol{X}] \geq R[\boldsymbol{Y}]$), concavity ($R$ is concave), translation-equivarience ($R[\boldsymbol{X}+c]=R[\boldsymbol{X}]+c$ with $c \in \mathbb{R}$) and positive homogeneity ($R[\lambda \boldsymbol{X}]=\lambda R[\boldsymbol{X}]$, with $\lambda \geq 0$). 

\subsection{Risk neutral formulation}\label{risk neutral}
\subsubsection{Centralized case}

First, we consider centralized case, where global Market Operator minimizes social cost for risk-neutral community. We denote feasibility sets as
\begin{multline}
    \mathcal{K}_n(\boldsymbol{x}_{-n}):= \{\boldsymbol{x}_n=(D_n, G_n, q_{n}) | D_n \in \mathcal{D}_n, G_n \in \mathcal{G}_n,  \\
    (\ref{reciprocity}), (\ref{eq:capacity}), (\ref{eq:supply_demand_eq}) \text{ hold } \forall \omega \in \Omega \},
\end{multline}
where $\boldsymbol{x}_{-n}$ denotes a vector of decisions of all agents excluding agent $n$, and joint admissible set as a $\mathcal{K} := \prod_{n} \mathcal{K}_n(\boldsymbol{x}_{-n})$. Then, formulation is as follows:
\begin{align*}
\min_{\bm{D}, \bm{G}, \bm{q}} \hspace{1cm} &   \sum_{\omega \in \Omega} p_{\omega}\sum_{n \in \mathcal{N}} \Pi_n (\omega),\\
s.t. \hspace{1cm} & (\bm{D}, \bm{G}, \bm{q}) \in \mathcal{K}.
\end{align*}

The objective function is convex as the sum of convex functions defined on a convex feasibility set. Indeed, the feasibility set is obtained as Cartesian product of convex sets. We can compute the Lagrangian function associated with the standard constrained optimization problem of social cost minimization under constraints and derive KKT conditions associated with it to determine the solution of the centralized market design optimization problem. We denote $\underline{\nu}^{\omega}_n, \overline{\nu}^{\omega}_n$ dual variables associated with constraint $\underline{G}_n \leq G_n(\omega) \leq \overline{G}_n$ respectively, $\underline{\mu}^{\omega}_n, \overline{\mu}^{\omega}_n$ dual variables associated with constraint $\underline{D}_n \leq D_n(\omega) \leq \overline{D}_n$. Let $\xi^{\omega}_n$ denote dual variable associated with constraint (\ref{eq:capacity}), $\zeta^{\omega}_n$ denote dual variable for (\ref{reciprocity}) and $\lambda^{\omega}_n$ for (\ref{eq:supply_demand_eq}).

Taking the derivative of the Lagrangian function with respect to $D_n(\omega)$, $G_n(\omega)$, $q_{mn}(\omega)$, for all $n \in \mathcal{N}$ and all $m \in \Gamma_n, m \not= n$, the stationarity conditions write down as follows.

Computations of the first order stationarity conditions give:
\begin{multline}\label{eq:sc1}
    \frac{\partial \mathcal{L}}{\partial D_n(\omega)} = 0 \Leftrightarrow 2 p_{\omega} \tilde{a}_n \Big(D_n (\omega) - D^{\ast}_n(\omega)\Big) \\
    - \underline{\mu}^{\omega}_n + \overline{\mu}^{\omega}_n + \lambda^{\omega}_n = 0, \forall \omega \in \Omega,
\end{multline}
\begin{multline}\label{eq:sc2}
    \frac{\partial \mathcal{L}}{\partial G_n(\omega)} = 0 \Leftrightarrow p_{\omega} \Big(a_n G_n(\omega) + b_n \Big) \\
    - \underline{\nu}^{\omega}_n + \overline{\nu}^{\omega}_n - \lambda^{\omega}_n = 0,\forall \omega \in \Omega,
\end{multline}
\begin{multline}\label{eq:sc3}
    \frac{\partial \mathcal{L}}{\partial q_{mn}(\omega)} = 0 \Leftrightarrow p_{\omega} \Big(2a_{mn}(q_{mn} + q_{nm})+ b_{mn}\Big)\\
     + p_{\omega}p_0(\omega) + \xi^{\omega}_{nm} + \zeta^{\omega}_{nm} -  \lambda^{\omega}_n = 0, \forall \omega \in \Omega. 
\end{multline}

Complementarity slackness conditions take the following form:
\begin{subequations}
\begin{gather}
    0 \leq \underline{\mu}^{\omega}_n \,\,\bot\,\, D_n(\omega) - \underline{D}_n\geq 0, \quad \forall n \in \mathcal{N}, \\
    0 \leq \overline{\mu}^{\omega}_n \,\,\bot\,\, \overline{D}_n - D_n(\omega)\geq 0, \quad \forall n \in \mathcal{N},\\
    0 \leq \underline{\nu}^{\omega}_n \,\,\bot\,\, G_n(\omega) - \underline{G}_n\geq 0,\quad \forall n \in \mathcal{N},\\
    0 \leq \overline{\nu}^{\omega}_n \,\,\bot\,\, \overline{G}_n - G_n(\omega)\geq 0, \quad \forall n \in \mathcal{N}, \\
    0 \leq \xi^{\omega}_{nm} \,\,\bot\,\, \kappa_{mn} - q_{mn} \geq 0, \forall m \in \Gamma_n, m \not= n, \forall n \in \mathcal{N}, \\
    0 \leq \zeta^{\omega}_{nm} \,\,\bot\,\, -q_{mn} - q_{nm} \geq 0, \forall m \in \Gamma_n, m > n, \forall n \in \mathcal{N}.
\end{gather}
\end{subequations}

From the last equation we infer that the nodal price at n can be expressed analytically as the trading decisions $q_{mn}$ and $q_{nm}$, the congestion constraint dual variable, the price of the import and the bilateral trade prices for all $\omega \in \Omega$:
\begin{equation}
    \lambda^{\omega}_n = p_{\omega} (2a_{mn}(q_{mn} + q_{nm})+b_{mn} + p_0(\omega)) + \xi^{\omega}_{nm} + \zeta^{\omega}_{nm}.
\end{equation}

From equation \eqref{eq:sc3}, we infer, for any couple of nodes $n \in \mathcal{N}, m \in \Gamma_n, m > n$, that $\forall \omega \in \Omega$:
\begin{gather*}
    \zeta^{\omega}_{mn} = \lambda^{\omega}_m - \xi^{\omega}_{mn} - p_{\omega} (2a_{mn}(q_{nm} + q_{mn})+b_{nm} + p_0(\omega)), \\
    \zeta^{\omega}_{nm} = \lambda^{\omega}_n - \xi^{\omega}_{nm} - p_{\omega} (2a_{mn}(q_{mn} + q_{nm})+b_{mn} + p_0(\omega)).
\end{gather*}

Subtracting those two last members in the equation above, we infer that $\forall m \in \Gamma_n, m \not= n, \forall n \in \mathcal{N}, \forall \omega \in \Omega$:
\begin{equation}
    \lambda^{\omega}_n - \lambda^{\omega}_m = \xi^{\omega}_{nm} - \xi^{\omega}_{mn} + p_{\omega}(b_{mn} - b_{nm}).
\end{equation}
Then, can we have a closed form expression for the nodal price as a function of $\lambda_0^{\omega}$ only. Indeed node 0 is assumed to belong to any neighborhood. We obtain the following expression for the nodal price at node $n$:
\begin{equation} \label{eq:lambda_n}
    \lambda_n^{\omega} = \lambda_0^{\omega} + \xi^{\omega}_{n0} - \xi^{\omega}_{0n} + p_{\omega}(b_{0n} - b_{n0}), \forall n \in \mathcal{N}, \forall \omega \in \Omega.
\end{equation}

From first two equation of KKT conditions, we infer that, at the optimum, for each node $n$:
\begin{gather}
    D_n(\omega) = D^{\ast}_n(\omega) - \frac{1}{2p_{\omega}\tilde{a}_n}(\lambda^{\omega}_n + (\overline{\mu}^{\omega}_n - \underline{\mu}^{\omega}_n)), \label{eq:D} \\
    G_n(\omega) = -\frac{b_n}{a_n} + \frac{1}{p_{\omega}a_n}(\lambda^{\omega}_n - (\overline{\nu}^{\omega}_n - \underline{\nu}^{\omega}_n)).
    \label{eq:G}
\end{gather}
From these expressions and (\ref{eq:supply_demand_eq}) we infer that the net import at node $n$ can be expressed as a linear function of the nodal price:
\begin{multline}\label{eq:Q}
    Q_n(\omega) = D^{\ast}_n - \frac{1}{2p_{\omega}\tilde{a}_n} (\overline{\mu}^{\omega}_n - \underline{\mu}^{\omega}_n) + \frac{b_n}{a_n} + \frac{1}{p_{\omega}a_n}(\overline{\nu}^{\omega}_n - \underline{\nu}^{\omega}_n) \\
    - (\frac{1}{2p_{\omega}\tilde{a}_n} + \frac{1}{p_{\omega}a_n})\lambda^{\omega}_n - \Delta G_n(\omega).
\end{multline}
The results are summarized in the following proposition.
\begin{proposition}
    At the optimum, the demand, flexibility activation and net import at each node $n$ can be expressed as linear functions of the nodal price at that node, given by equations~\eqref{eq:D}, \eqref{eq:G}, \eqref{eq:Q}.
\end{proposition}

The total sum of the net imports at all nodes should be negative or null, i.e., $\sum_{n \in \mathcal{N}} Q_n(\omega) \leq 0, \forall \omega \in \Omega$. From the supply-demand balancing, this is equivalent to $\sum_{n \in \mathcal{N}} (D_n(\omega) - G_n(\omega)) \leq \sum_{n \in \mathcal{N}} \Delta G_n (\omega)$.  A strict inequality would lead to a situation where a part of the energy generation is wasted. That is not acceptable. To avoid that situation, the RES-based generation should be limited and the demand capacities large enough. Note that the sizing of the prosumers’ capacities and RES-based generation possible clipping strategies are out of the scope of this work. This result is formalized in the proposition below. 

\begin{proposition} \label{prop:CN_no_surplus}
    A necessary condition for no energy surplus is that there is at least one prosumer $n$ in $\mathcal{N}$ whose capacities and RES-based generation are such that $D_n(\omega) - G_n(\omega) \geq \Delta G_n(\omega)$.
\end{proposition}

\begin{proof}
    Proof coincides with the proof of Proposition 2 of \cite{lecadre}.
\end{proof}

Substituting \eqref{eq:lambda_n} in \eqref{eq:Q} and taking the sum over $n$, can we derive that:

\begin{multline}
    \sum_{n \in \mathcal{N}} Q_n \leq 0 \Leftrightarrow \lambda^{\omega}_0 \sum_{n} \left( \frac{1}{2p_{\omega} \tilde{a}_n} + \frac{1}{p_{\omega} a_n} \right) \geq  \\ 
    \geq \sum_{n} \Biggl( \left( D^{\ast}_n - \frac{1}{2p_{\omega}\tilde{a}_n} (\overline{\mu}^{\omega}_n - \underline{\mu}^{\omega}_n) + \frac{b_n}{a_n} + \frac{1}{p_{\omega}a_n}(\overline{\nu}^{\omega}_n - \underline{\nu}^{\omega}_n) \right) -\\
    - \left( \frac{1}{2p_{\omega} \tilde{a}_n} + \frac{1}{p_{\omega} a_n} \right) (\xi^{\omega}_{n0} - \xi^{\omega}_{0n} + p_{\omega}(b_{0n} - b_{n0})) - \Delta G_n \Biggr),
\end{multline}
which gives a necessary and sufficient condition on $\lambda^{\omega}_0$ to avoid energy surplus.

\subsubsection{Peer-to-Peer case}
In the peer-to-peer setting, each agent $n \in \mathcal{N}$ determines, by herself, her demand, energy generation and bilateral trades with other agents in her local energy community under constraints on demand, flexibility activation and transmission capacity so as to minimize her expected costs.

Formally, each agent in node $n \in \mathcal{N}$ solves the following optimization problem:
\begin{subequations}
\label{problem_user_neutral}
\begin{align}
\min_{D_n, G_n, \bm{q}_n} \hspace{1cm} &  \E [\Pi_n], \\
s.t. \hspace{1cm} & (D_n, G_n, \bm{q}_n) \in \mathcal{K}_n.
\end{align}
\end{subequations}

\begin{remark} \label{remark:KKTp2p}
All the KKT conditions except \eqref{eq:sc3} coincide with centralized case. Under peer-to-peer market design, equation \eqref{eq:sc3} will be:
\begin{multline}\label{dl/dq_p2p_neutral}
    \frac{\partial \mathcal{L}}{\partial q_{mn}} = 0 \Leftrightarrow p_{\omega} \Big(2a_{mn}q_{mn} + a_{mn}q_{nm}+b_{mn} + p_0(\omega)\Big)+
    \\ + \xi^{\omega}_{nm} + \zeta^{\omega}_{nm} -  \lambda^{\omega}_n = 0, \forall \omega \in \Omega.
\end{multline}
\end{remark}
From the last equation we infer that the nodal price at $n$ can be expressed analytically as:
\begin{equation*}
    \lambda^{\omega}_n = p_{\omega} \Big(2a_{mn}q_{mn} + a_{mn}q_{nm}+b_{mn} + p_0(\omega)\Big) + \xi^{\omega}_{nm} + \zeta^{\omega}_{nm}.
\end{equation*}
From (\ref{dl/dq_p2p_neutral}) following the same path as in centralized case we obtain a closed form expression for the nodal price as a function of $\lambda_0^{\omega}$:
\begin{equation}\label{lambda_n_p2p_neutral}
    \lambda_n^{\omega} = \lambda_0^{\omega}+p_{\omega}a_{n0}(q_{0n}-q_{n0})+p_{\omega}(b_{0n}-b_{n0})+\xi_{n0}^{\omega}-\xi_{0n}^{\omega}.
\end{equation}

\begin{remark}
    From the coincidence of KKT conditions with centralized case, we infer that expressions for $D_n, G_n, Q_n$ at the optimum coincide with (\ref{eq:D}), (\ref{eq:G}), (\ref{eq:Q}) respectively. However, note, that $\lambda^{\omega}_n$ appearing in these expressions differ from the centralized case.
\end{remark}

\begin{remark} \label{risk neutral lambda comoparison}
    Note that expressions for $\lambda_n$ in centralized (\ref{eq:lambda_n}) and peer-to-peer (\ref{lambda_n_p2p_neutral}) settings differ only by the term $p_{\omega}a_{n0}(q_{0n}-q_{n0})$. In the case $a_{0n} \geq 0$ the nodal price in the peer-to-peer case is smaller than in the centralized $\Leftrightarrow$ $q_{0n} < q_{n0}$ -- node $n$ "exports" more energy to node 0, than "imports" from it. Note, that if "export" between node $n$ and node 0 equals "import", then (\ref{lambda_n_p2p_neutral}) coincides  with (\ref{eq:lambda_n}). 
    
    This situation may be viewed as a case of isolated community where there is no possibility to trade energy ($q_{0n} = q_{n0} = 0$) with the grid (e.g. islanded microgrid). In this setting peer-to-peer framework coincides with centralized market design.
\end{remark}

\subsection{Risk Averse Market Design}\label{risk averse}

\subsubsection{Centralized case}
The centralized market design is inspired from the existing pool-based markets. The global Market Operator (MO) maximizes the social welfare defined in (\ref{eq:social_welfare}) in the risk-averse agents framework:
\begin{subequations}
\label{eq:risk_averse_centralized}
\begin{align}
\min_{\bm{D}, \bm{G}, \bm{q}} \hspace{1cm} &  R [SC], \nonumber \\
s.t. \hspace{1cm} & (\bm{D}, \bm{G}, \bm{q}) \in \mathcal{K}.
\end{align}
\end{subequations}

For the case of our study we implement the cVaR as a coherent risk measure.  Denote $\chi_n$ as a risk attitude of the prosumer $n$. The cVaR is the average of all realizations larger than the VaR, where $VaR_n = \min_{\eta_n} \{ \eta_n \, |\, \mathbb{P} [\Pi^{\omega}_n \leq \eta_n] = \chi_n \}$. With a little abuse of notations we denote $VaR_n$ as $\eta_n$. Then, following \cite{rockafellar}, we write cVaR as follows:
\begin{equation}
    R [\Pi_n(\omega)] =  \eta_n + \frac{1}{ (1-\chi_n)} \sum_{\omega \in \Omega} p_{\omega}[\Pi_n(\omega) - \eta_n]^{+}. 
\end{equation}
Note, that $R[\Pi_n(\omega)]$ is convex in $(D_n, G_n, \bm{q_n}, \eta_n)$ if $\Pi_n(\omega)$ is convex in $(D_n, G_n, \bm{q_n})$, which is the case in our model. The non-differentiability of $R[\Pi_n(\omega)]$ can be overcome by leveraging the epigraph form of it as in \cite{rockafellar}:
\begin{equation}
    R [\Pi_n(\omega)] =  \eta_n + \frac{1}{ (1-\chi_n)} \sum_{\omega \in \Omega}p_{\omega} u^{\omega}_n,
\end{equation}
with $u^{\omega}_n \geq 0$  and $\Pi_n(\omega) - \eta_n \leq u^{\omega}_n$ with dual variables $\pi^{\omega}_n$ and $\tau^{\omega}_n$ respectively. Define feasibility set $\tilde{\mathcal{K}}_n$ as
\begin{multline}
    \tilde{\mathcal{K}}_n(\boldsymbol{x}_{-n}) := \{\boldsymbol{x}_n=(D_n, G_n, q_{n}, u_n) | (D_n, G_n, q_{n}) \\ \in \mathcal{K}_n(\boldsymbol{x}_{-n}), 
    u^{\omega}_n \geq 0, \Pi_n(\omega) - \eta_n \leq u^{\omega}_n \},
\end{multline}
and again $\tilde{\mathcal{K}} := \prod_n \tilde{\mathcal{K}}_n(\boldsymbol{x}_{-n})$.
Then, risk-averse formulation of the centralized market design takes the following form:
\begin{subequations} \label{eq:reformulation_risk_averse}
\begin{align}
\min_{\bm{D}, \bm{G}, \bm{q}, \bm{u^{\omega}}} \hspace{1cm} & \sum_{n \in \mathcal{N}} \left( \eta_n + \frac{1}{ (1-\chi_n)} \sum_{\omega \in \Omega}p_{\omega} u^{\omega}_n \right), \\
s.t. \hspace{1cm} & (\bm{D}, \bm{G}, \bm{q}, \bm{u}) \in \tilde{\mathcal{K}}.
\end{align}
\end{subequations}

\begin{proposition}\label{equivalence of the epigraph form}
    The reformulation \eqref{eq:reformulation_risk_averse} is equivalent to the optimization problem~\eqref{eq:risk_averse_centralized}.
\end{proposition}
\begin{proof}
    Assume $(D^{\ast}_n, G^{\ast}_n, (\bm{q}^{\ast}_n)_n)$ is an optimum solution of \eqref{eq:risk_averse_centralized}. Denote $S^{\ast}$ as an optimal value of the problem \eqref{eq:risk_averse_centralized} and $\Pi^{\ast}_n(\omega)$ - value of utility function of agent $n$ in the point $(D^{\ast}_n, G^{\ast}_n, (\bm{q}^{\ast}_n)_n)$. Consider the following problem:
    \begin{subequations}
    \begin{align}\label{equivalence_proof}
        \min_{u^{\omega}_n} &\sum_{n \in \mathcal{N}} \left(\eta_n + \frac{1}{(1-\chi_n)}\sum_{\omega \in \Omega} p_{\omega}u^{\omega}_n \right), \\
        s.t. \hspace{1cm} & u^{\omega}_n \geq \Pi^{\ast}_n(\omega) - \eta_n, \\
        &u^{\omega}_n \geq 0.
    \end{align}
    \end{subequations}
    Denote $S^{\ast \ast}$ as an optimal value of the above problem. It is clear that for each fixed $(D^{\ast}_n, G^{\ast}_n, (\bm{q}^{\ast}_n)_n)$, $S^{\ast} = S^{\ast \ast}$.
\end{proof}

The objective function is convex as the sum of convex functions defined on a convex feasibility set. Computations of the first order stationarity conditions of the KKT give:
\begin{multline}
    \frac{\partial \mathcal{L}}{\partial D_n(\omega)} = 0 \Leftrightarrow  2 \tau^{\omega}_{n} \tilde{a}_n(D_n(\omega)-D^{\ast}_n) \\
    -\underline{\mu}^{\omega}_n + \overline{\mu}^{\omega}_n + \lambda^{\omega}_n = 0,
\end{multline}
\begin{multline}
    \frac{\partial \mathcal{L}}{\partial G_n(\omega)} = 0 \Leftrightarrow \tau^{\omega}_{n} a_n ( G_n(\omega) +b_n) - \underline{\nu}^{\omega}_n \\
    + \overline{\nu}^{\omega}_n - \lambda^{\omega}_n = 0, 
\end{multline}
\begin{multline}\label{dl/dq}
    \frac{\partial \mathcal{L}}{\partial q_{mn}(\omega)} = 0 \Leftrightarrow \zeta^{\omega}_{nm}+\xi^{\omega}_{nm} +  \tau^{\omega}_{n} \left( 2a_{mn} q_{mn} +a_{mn}q_{nm} \right) \\ +\tau^{\omega}_{n}\left(b_{mn}+p_0(\omega) \right)  + \tau^{\omega}_{m} a_{nm} q_{nm} - \lambda^{\omega}_n =0, 
\end{multline}
\begin{equation}
    \frac{\partial \mathcal{L}}{\partial u^{\omega}_n} =0 \Leftrightarrow \frac{p_{\omega}}{1-\chi_n} - \tau^{\omega}_{n} - \pi^{\omega}_n = 0. 
\end{equation}

\begin{remark}\label{subjective probabilities definition}
    It is easy to see that expressions for $D_n(\omega), G_n(\omega), Q_n(\omega)$ coincide with (\ref{eq:D}), (\ref{eq:G}), (\ref{eq:Q}) with two differences: first, instead of $p_{\omega}$ we have $\tau^{\omega}_{n}$, which is expressed as: 
    \begin{equation}
        \tau^{\omega}_n = \frac{p_{\omega}}{1-\chi_n} - \pi^{\omega}_n. \label{eq:tau}
    \end{equation}
    Second, expression for $\lambda$ differs (because of $\tau^{\omega}_{n}$ as well). $\tau^{\omega}_n$ could be viewed as a risk-adjusted probabilities \cite{moret} from agent $n$'s point of view.

Note that there are two additional complementarity slackness conditions (comparing to the risk neutral case) for constraints $u^{\omega}_n \geq 0$ and $ \Pi_n(\omega) - \eta_n \leq u^{\omega}_n$ :
\begin{subequations}
\begin{gather}
    0 \leq \pi^{\omega}_n \,\,\bot\,\, u^{\omega}_n \geq 0, \label{comp_pi}\\
    0 \leq \tau^{\omega}_n \,\,\bot\,\, u^{\omega}_n - (\Pi_n(\omega) - \eta_n) \geq 0. \label{comp_tau}
\end{gather}
\end{subequations}
It follows that we can consider two complementary cases:
\begin{enumerate}
    \item $\Pi_{n}(\omega) - \eta_n \geq 0$ -- it means that $u^{\omega}_n \geq 0$ and from (\ref{comp_pi}) $\pi^{\omega}_n = 0$. It means that when loss function $\Pi_n(\omega)$ of agent $n$ exceeds some threshold $\eta_n$ defined through his risk attitude $\chi_n$, dual variable $\tau^{\omega}_n$ is defined as:
    \begin{equation}
        \tau^{\omega}_n = \frac{p_{\omega}}{1-\chi_n}.
    \end{equation}
    
    \item $\Pi_{n}(\omega) - \eta_n < 0$ -- it means that $u^{\omega}_n = 0$ and from (\ref{comp_tau}) $\tau^{\omega}_n = 0$. It means that when loss function $\Pi_n(\omega)$ of agent $n$ is below $\eta_n$, dual variable $\pi^{\omega}_n$ is defined as: 
    \begin{equation}
        \pi^{\omega}_n = \frac{p_{\omega}}{1-\chi_n}.
    \end{equation}
\end{enumerate}

\end{remark}

\subsubsection{Peer-to-Peer Case}

In the peer-to-peer setting, each agent $n \in \mathcal{N}$ determines, by herself, her demand, energy generation and bilateral trades with other agents in her local energy community under constraints on demand, flexibility activation and transmission capacity so as to minimize her costs. Agent $n$'s optimization problem takes the following form:
\begin{subequations}
\label{problemUsern}
\begin{align}
\min_{D_n, G_n, \bm{q}_n} \hspace{1cm} & R_n[\Pi_n], \label{eq:p2p_objectif} \\
s.t. \hspace{1cm} & (D_n, G_n, \bm{q}_n) \in \mathcal{K}_n(\boldsymbol{x}_{-n}),
\end{align}
\end{subequations}
where $\bm{q}_n= \big(q_{mn}\big)_{m\in\Gamma_n}$ are the trading decisions of prosumer $n$.

Again implementing cVaR as a risk measure and leveraging its epigraph form, we formulate agent $n$'s optimization problem as:

\begin{subequations}
\begin{align}
\min_{D_n, G_n, \bm{q}_n, u_n} \hspace{1cm} &   \eta_n + \frac{1}{ (1-\chi_n)} \sum_{\omega \in \Omega} p_{\omega} u^{\omega}_n,\\
s.t.  \hspace{1cm} &(D_n, G_n, \bm{q}_n, u_n) \in \tilde{\mathcal{K}}_n(\boldsymbol{x}_{-n}).
\end{align}
\end{subequations}

\begin{remark}
    Note that all the KKT conditions except (\ref{dl/dq}) coincide with the centralized case. In peer-to-peer setting, (\ref{dl/dq}) will be:
    \begin{multline}
        \frac{\partial \mathcal{L}}{\partial q_{mn}(\omega)} = 0 \Leftrightarrow \zeta^{\omega}_{nm}+\xi^{\omega}_{nm} + \tau^{\omega}_{n} \left( 2a_{mn} q_{mn} +a_{mn}q_{nm} \right)\\ +\tau^{\omega}_n\left(b_{mn}+p_0(\omega) \right) - \lambda^{\omega}_n =0.
    \end{multline}
\end{remark}

\section{Equilibrium problems}\label{equilibrium problems}

We consider equilibrium problems in an incomplete market.

\subsection{Basic definitions}

Since our problem is formulated with a scenario based approach, the classical definitions of equilibrium (GNE, VE) introduced in a deterministic setting are easily transposable for each scenario.

\begin{definition}{\emph{Generalized Nash equilibrium} \cite{kulkarni}}
A Generalized Nash Equilibrium (GNE) of the game defined by the minimization problems \eqref{problem_user_neutral} (resp. \eqref{problemUsern}) with coupling constraints, is a vector $(D_n(\omega), G_n(\omega), \bm{q}(\omega)_n)_n$ that solves the minimization problems \eqref{problem_user_neutral} (resp. \eqref{problemUsern}) or, equivalently, a vector $(D_n(\omega), G_n(\omega), \bm{q}(\omega)_n)_n$ such that $(D_n(\omega), G_n(\omega), \bm{q}(\omega)_n)_n$ solves the system $KKT_n$ for each $n$.

\end{definition}

\begin{definition}{\emph{Variational equilibrium} \cite{kulkarni}}
A Variational Equilibrium (VE) of the game defined by the minimization problems \eqref{problem_user_neutral} (resp. \eqref{problemUsern}) with coupling constraints, is a vector $(D_n(\omega), G_n(\omega), \bm{q}(\omega)_n)_n$ that solves the minimization problems \eqref{problem_user_neutral} (resp. \eqref{problemUsern}) or, equivalently, a vector $(D_n(\omega), G_n(\omega), \bm{q}(\omega)_n)_n$ such that $(D_n(\omega), G_n(\omega), \bm{q}(\omega)_n)_n$ solves the system $KKT_n$ for each $n$ and, in addition, such that the Lagrangian multipliers associated to the coupling constraints (\ref{reciprocity}) are equal, i.e.:
\begin{equation}
    \zeta^{\omega}_{nm} = \zeta^{\omega}_{mn}.
\end{equation}
\end{definition}
The term ``variational'' refers to the variational inequality problem associated to such an equilibrium: $\bm{\hat{x}} \in \mathcal{K}$ is a Variational Equilibrium if, and only if, it is a solution of:
\begin{equation}
    \left\langle \nabla \Pi_n(\bm{\hat{x}}), \bm{x} - \bm{\hat{x}}\right\rangle \leq 0, \forall x \in \mathcal{K}.
\end{equation}

\subsection{Results}

We can observe, that in risk averse case in the setting of Example~\ref{ex1}, Variational Equilibria are defined by exactly the same KKTs system than the social cost minimizer, which leads to the following result:
\begin{proposition}\label{simplified VE coincidence for risk averse}
    Assuming $a_{mn}=a_{nm}=0, \forall m,n\in\mc{N}, m\not= n$, in the risk averse case the set of Variational Equilibria (such that $\zeta^{\omega}_{nm} = \zeta^{\omega}_{mn}$ for all $n \in \mathcal{N}$ and all $m \not= n \in \Gamma_n$) coincides with the set of social cost optima.
\end{proposition}

\begin{remark}
    For risk-neutral case with the general formulation of the trading cost $\tilde{C}_{n}(\bm{q}_n)$ it doesn't hold as is shown in Remark~\ref{remark:KKTp2p}.
\end{remark}

\begin{proposition}\label{general VE coincidence for risk neutral}
    Risk-neutral peer-to-peer market design problem admits a unique Variational Equilibrium.
\end{proposition}
\begin{proof}
    Note, that in the setting of risk-neutral case, utility function of player $n$ can be viewed as the sum of two terms: 
    \begin{itemize}
        \item $C_n(G_n(\omega)) + p_0(\omega)Q_n(\omega) - U_n(D_n(\omega))$ -- user-specific part of utility function, which depends only on the variables of player $n$ only and not on the variables of other player. Note, that it is strongly convex twice continuously differentiable function as the sum of strongly convex and convex functions of $x_n = (D_n, G_n, \bm{q}_n)$. 
        
        \item $\tilde{C}_n(\bm{q}_n(\omega))$ -- trading cost of agent $n$, modeled as a function of trading decisions of \textit{agents} on the edge $(n,m)$. Note, that trading decisions $q_{mn}$ of agent $n$ could be viewed as a decision of agent $n$ on the link $(m,n)$ - that is, let $\mathcal{L}$ denote the set of all edges of the network. Then, without loss of generality and a small abuse of notation, we denote function $q^{n}_{l}, l = (m,n), \forall n,m \in \mathcal{N}$  as follows:
        \begin{equation*}
            q^n_l = \left\{ 
            \begin{aligned}
                & q_{mn}, \quad m \in \Gamma_n, \\
                & 0, \quad \text{otherwise}.
            \end{aligned}
            \right.
        \end{equation*}
    \end{itemize}
    Then, $c_{nm}(\bm{q})$ - link-specific cost - can be rewritten for link $l \in \mathcal{L}$ as:
    \begin{equation*}
        c_{nm}(\bm{q}) = c_l(\bm{q}) = a_{l} \sum_{n \in \mathcal{N}} q^{n}_l + b_l.
    \end{equation*}
    Note that $(c_l)'_{i} = (c_l)'_{j} > 0, i \not= j$ and $(c_l)''_{ij} = 0$. To conclude, note, that strategy sets of each user are closed and convex. Then, from Proposition 1 in \cite{yin}, it follows that peer-to-peer problem admits unique VE.
\end{proof}

\section{Completeness of the market}\label{completeness of the market}
To complete the market in the sense of definition given in \cite{baron}, \cite{moret} - that is - a market is said to be complete whenever there exists an equilibrium price for every asset in every possible state of the world. The market is incomplete otherwise., we include financial contracts that are intended to hedge the risk for market participants \cite{gerard, moret}. We assume that the agents can trade risk with each other using these contracts, i.e., they
pay a certain amount contingent on a given scenario occurring. Note, that risk trading layer differs from energy trading layer (e.g. it doesn't depend on physical properties) - in this section we assume that agent $n$ can trade risk with the whole community $\mathcal{N}$. The price for the contract corresponding to the scenario $\omega$ is denoted $\alpha^{\omega}$. We also assume existence of the exogenous agent (e.g. aggregator, insurance company etc.) which back up shortfall of the risk traded with fixed prices $\gamma^{\omega}$ per scenario $\omega$. We impose that for each $\omega \in \Omega$, $\alpha^{\omega} \leq 1, \gamma^{\omega} \leq 1$ so it is beneficial to buy contract corresponding to the scenario $\omega$.

Trading risk involves prosumers making a first-stage decision which consists in agent $n$ paying $\sum_{\omega \in \Omega} \alpha^{\omega} W^{\omega}_n + \sum_{\omega \in \Omega} \gamma^{\omega} J^{\omega}_n$ to obtain contingent payments $W^{\omega}_n + J^{\omega}_n$ in each outcome $\omega \in \Omega$. Here we assume that $W^{\omega}_n$ are the contracts traded inside the community and $J^{\omega}_n$ are the contracts traded with an exogenous agent. Note that $W^{\omega}_n$ could be negative, while $J^{\omega}_n \geq 0$.

Then, objective function $R_n[\Pi_n]$ of prosumer $n$ in risk-averse case is expressed as follows:
\begin{multline}
     R_n [\Pi_n(\omega)] =  \eta_n + \sum_{\omega \in \Omega} \alpha^{\omega} W^{\omega}_n + \sum_{\omega \in \Omega} \gamma^{\omega} J^{\omega}_n \\
     + \frac{1}{ (1-\chi_n)} \sum_{\omega \in \Omega} p_{\omega}[\Pi_n(\omega) - (W^{\omega}_n + J^{\omega}_n) - \eta_n]^{+}.     
\end{multline}
We add risk trading balance per scenario condition:
\begin{equation}\label{risk balance condition}
    \sum_{n \in \mathcal{N}} \left( W^{\omega}_n + J^{\omega}_n \right)  = 0, \qquad \forall \omega \in \Omega,
\end{equation}
with dual variable $\phi^{\omega}$ associated to it.
Define feasibility set $\hat{\mathcal{K}}$ as:
\begin{multline}
    \hat{\mathcal{K}}_n(\boldsymbol{x}_{-n}) := \{\boldsymbol{x}_n=(D_n, G_n,q_{n}, u_n, W_n, J_n) | (D_n, G_n,\\ q_{n}, u_n) \in \tilde{\mathcal{K}}_n(\boldsymbol{x}_{-n}), 
    (\ref{risk balance condition}) \text{ and } 0 \leq J^{\omega}_n \text{ hold } \forall \omega \in \Omega   \},
\end{multline}
and $\hat{\mathcal{K}} := \prod_n \hat{\mathcal{K}}_n(\boldsymbol{x}_{-n})$.
\subsection{Centralized case}
Thus, we formulate risk-averse formulation with risk trading using auxiliary variables $u^{\omega}_n$ as:
\begin{subequations} \label{eq:reformulation_risk_trading_centralized}
\begin{align}
\min_{\bm{D}, \bm{G}, \bm{q}, \bm{u^{\omega}}, \bm{W^{\omega}}, \bm{J^{\omega}}} \hspace{1cm} & \sum_{n \in \mathcal{N}} R_n[\Pi_n(\omega)],\\
s.t. \hspace{1cm} & (\bm{D}, \bm{G}, \bm{q}, \bm{u}, \bm{W}, \bm{J}) \in \hat{\mathcal{K}}.
\end{align}
\end{subequations}

Comparing to the model without financial contracts, there will be two additional first order stationarity conditions:
\begin{subequations}
\begin{gather}
    \frac{\partial \mathcal{L}}{\partial W^{\omega}_n} = 0 \Leftrightarrow  \alpha^{\omega} - \tau^{\omega}_n + \phi^{\omega} = 0, \label{dl/dw}\\
    \frac{\partial \mathcal{L}}{\partial J^{\omega}_n} = 0 \Leftrightarrow  \gamma^{\omega} - \tau^{\omega}_n + \phi^{\omega}  - \sigma^{\omega}_n= 0. \label{dl/dj}
\end{gather}
\end{subequations}

\begin{theorem}\label{chracterization of the risk trading price}
    Price $\alpha^{\omega}$ of the contracts inside the community is not bigger than the exogenous price $\gamma^{\omega}$. More precisely, there are three regimes of the market:
    
    \textbf{(1)} When exogenous company settles trading price $\gamma^{\omega} \leq \min_{n \in \mathcal{N}} \frac{p_{\omega}}{1-\chi_n}$, price $\alpha^{\omega}$ for trading contracts inside the community is defined as follows:
    \begin{itemize}
    \item[] \textbf{(1.a)} When there is at least one agent trading risk with the exogenous agent, $\alpha^{\omega} = \gamma^{\omega}$. 
            
    \item[] (\textbf{1.b}) When the risk is fully traded inside the community ($J^{\omega}_n = 0, \, \forall n \in \mathcal{N}$), $\alpha^{\omega} = -\phi^{\omega} \leq \gamma^{\omega}$, where $\phi^{\omega}$ is dual variable for risk balance condition.
    \end{itemize}
    
    \textbf{(2)} When exogenous price is $\gamma^{\omega} \geq \max_{n} \frac{p_{\omega}}{1-\chi_n}$, there are three possible regimes for price $\alpha^{\omega}$:
    \begin{itemize}
    \item[] \textbf{(2.a)} $\alpha^{\omega} = -\phi^{\omega}$, where $\phi^{\omega}$ is a dual variable for risk balance condition.
            
     \item[] \textbf{(2.b)} $\alpha^{\omega} = \tau^{\omega} - \phi^{\omega}$, where $\tau^{\omega}_n$ are aligned across agents. 
            
     \item[] \textbf{(2.c)} $\alpha^{\omega} = \frac{p_{\omega}}{1-\min_{n}\chi_n}$ -- price of the contracts for risk trading inside the community is equal to the least risk-averse prosumer's risk adjusted probability. 
     \end{itemize}   
    \textbf{(3)} When $\min_{n} \frac{p_{\omega}}{1-\chi_n} < \gamma^{\omega} < \max_{n} \frac{p_{\omega}}{1-\chi_n}$, contracts inside the community are settled as follow:
    \begin{itemize}
     \item[]\textbf{(3.a)} If there is at least one agent $n'$ with $\frac{p_{\omega}}{1-\chi_{n'}} < \gamma^{\omega}$ and $\Pi(\omega)_{n'} - \eta_{n'} > 0$, then $\alpha^{\omega} = \gamma^{\omega}$.
            
     \item[]\textbf{(3.b)} Otherwise, price $\alpha^{\omega}$ is settled as in the case (2).
     \end{itemize}
\end{theorem}
\begin{proof}
    From (\ref{dl/dw}) and (\ref{dl/dj}) we infer that $\alpha^{\omega} = \gamma^{\omega} - \sigma^{\omega}_n$. Note that from complementarity condition for constraint $J^{\omega}_n \geq 0$ we get $\sigma^{\omega}_n \geq 0$ which shows $\alpha^{\omega} \leq \gamma^{\omega}$. Also, note that in equation $\sigma^{\omega}_n = \gamma^{\omega} - \alpha^{\omega}$, right part doesn't depend on $n$, which means that $\sigma^{\omega}_n = \sigma^{\omega}_m, \, \forall m,n \in \mathcal{N}$. 
    \begin{itemize}
                \item[] \textbf{(1.a)} If at least one agent $n'$ trades risk with exogenous agent, that is $J^{\omega}_{n'} > 0$, then from complementarity conditions, $\sigma^{\omega}_n = \sigma^{\omega}_{n'} = 0, \, \forall n \in \mathcal{N}$. It follows directly, that $\alpha^{\omega} = \gamma^{\omega}$.
                
                \item[] \textbf{(1.b)} If $\forall n \in \mathcal{N}, \, J^{\omega}_n = 0$, then $\sigma^{\omega}_n \geq 0$. From (\ref{dl/dj}) we get 
                \begin{equation}
                    \sigma^{\omega} = \gamma^{\omega} + \phi^{\omega} - \tau^{\omega}_n,
                \end{equation}
                where we omit index $n$ in ${\sigma^{\omega}_n}$ for convenience. Note, that it means that $\tau^{\omega}_n$ are aligned across agents. We state that $\Pi_n(\omega) - W^{\omega}_n - \eta_n \leq 0,\, \forall n$, and thus from complementarity condition, $\tau^{\omega}_n = 0,  \forall n \in \mathcal{N}$. Indeed, from the contradiction: if there was an agent $n'$ with $\tau^{\omega}_{n'} \not= 0$, that is $\Pi_{n'}(\omega) - W^{\omega}_{n'} - \eta_{n'} > 0$, then she would buy contract $J^{\omega}_{n'}$ for the price $\gamma^{\omega} J^{\omega}_{n'}$, which is less than the gain from buying this contract: $\frac{p_{\omega}}{1-\chi_{n'}} J^{\omega}_{n'}$ and have a guaranteed decrease in cost function, which contradicts assumption $J^{\omega}_n =0, \forall n \in \mathcal{N}$.
            \end{itemize}
        
        \textbf{(2)} Note that when $\gamma^{\omega} \geq \max_{n} \frac{p_{\omega}}{1-\chi_n}$, it is non-profitable for any prosumer in the community to trade risk with the exogenous agent, which means that $J^{\omega}_n = 0, \,\forall n \in \mathcal{N}$. Similarly to the case (1.b) we get $\sigma^{\omega} = \gamma^{\omega} + \phi^{\omega} - \tau^{\omega}_n$ and $\tau^{\omega}_n = \tau^{\omega}_m,\,\forall n,m, \in \mathcal{N}$.
        
        \textbf{(2.a)} This price corresponds to the case when supply of the risk to trade inside the community exceeds the demand, that is, $\Pi_n(\omega) - W^{\omega}_n - \eta_n \leq 0,\, \forall n$. Similarly to (1.b), from complementarity conditions, $\tau^{\omega}_n = 0, \, \forall n$, and it follows that $\alpha^{\omega} = - \phi^{\omega}$.
            
        \textbf{(2.b)} Assume that in the equilibrium, $\Pi_n(\omega) - W^{\omega}_n - \eta_n \geq 0, \forall n \in \mathcal{N}$ and there exists at least one prosumer with $\Pi_{n'}(\omega) - W^{\omega}_{n'} - \eta_{n'} > 0$ -- that is $\pi^{\omega}_{n'} = 0$ for some $n' \in \mathcal{N}$. Denote $\mathcal{M} \subseteq \mathcal{N}$ - set of such agents. We get that if $n_1, n_2 \in \mathcal{M}$, then $\chi_{n_1} = \chi_{n_2}$ and if $n_1 \in \mathcal{M}, n_2 \not\in \mathcal{M}$, then 
        \begin{equation}
            \pi^{\omega}_{n_2} = p_{\omega} \frac{\chi_{n_2} - \chi_{n_1}}{(1-\chi_{n_2})(1-\chi_{n_1})}.
        \end{equation}
        Note that from complementarity condition $\pi^{\omega}_n \geq 0$. It follows that $\chi_{n_2} \geq \chi_{n_1}$, so, if $n_1 \in \mathcal{M}$ then $\chi_{n_1} = \min_{n} \chi_n$. If $\mathcal{M} = \mathcal{N}$, then $\Pi_n(\omega) - W^{\omega}_n - \eta_n > 0, \, \forall n$ - it means that $\pi^{\omega}_n = 0,\, \forall n \in \mathcal{N}$. It follows that $\frac{p_{\omega}}{1-\chi_{n_1}} = \frac{p_{\omega}}{1-\chi_{n_2}}, \forall n_1, n_2 \in \mathcal{N}$, so we are in homogeneous risk aversion case. 
            
        \textbf{(2.c)} This regime corresponds to the case when $\Pi_n(\omega) - W^{\omega}_n - \eta_n = 0, \, \forall n \in \mathcal{N}$. Define $\mathcal{B} \subset \mathcal{N}$ ($\mathcal{S} \subset \mathcal{N}$) - set of players who buy (sell) contracts in the optimum, that is - $\Pi(\omega)_n - \eta_n > 0$ ($\Pi(\omega)_n - \eta_n < 0$ ). Assume that  $\alpha^{\omega} \not= \frac{p_{\omega}}{1-\min_{n}\chi_n}$, then:
        \begin{itemize}
            \item If prosumer $n'$ is in $ \mathcal{B}$, where $n' := \argmin_n \frac{1}{1-\chi_n}$, it would be unprofitable for her to buy contracts $W^{\omega}_{n'}$, which contradicts with the assumption $\Pi_n(\omega) - W^{\omega}_n - \eta_n = 0, \, \forall n \in \mathcal{N}$. 
            
            \item If $n' \not\in \mathcal{S} \bigcup \mathcal{B}$, which means that $\Pi(\omega)_{n'} = \eta_{n'}$, then she could start selling contracts to decrease her cost $R[\Pi(\omega)_{n'}]$, thus leading to a contradiction.
            
            \item Similarly, the same contradiction arises in $n' \in \mathcal{S}$, when $|S| \geq 2$. 
        \end{itemize}

    \textbf{(3.a)} If there exists an agent $n'$ with $\frac{p_{\omega}}{1-\chi_{n'}} < \gamma^{\omega}$ and $\Pi_{n'}(\omega) - \eta_{n'} > 0$, then it would be profitable for her to buy contract from exogenous agent - that is, $J^{\omega}_{n'} > 0$. Then, proof coincides with the case (1.a).
    
    \textbf{(3.b)} Indeed, if there is no such agent, it means that $J^{\omega}_n = 0,\, \forall n \in \mathcal{N}$, which leads to the same framework as in (2).

\end{proof}

\begin{corollary}\label{subjective probabilities alignment}
    As a follow up of Proposition \ref{chracterization of the risk trading price}, in a complete market, the risk adjusted probabilities of the agents are aligned, i.e., 
    $\tau^{\omega}_n = \tau^{\omega}_{n'}, \forall n, n'\in \mathcal{N}$. 
\end{corollary}
\subsection{Decentralized case}
Decentralized setting for risk trading in the peer-to-peer framework is rarely considered in the literature for energy trading. Risk-averse formulation with risk trading contracts and auxiliary variables $u^{\omega}_n$ takes the following form:
\begin{subequations} \label{eq:reformulation_risk_trading_decentralized}
\begin{align}
\min_{\bm{D}, \bm{G}, \bm{q}, \bm{u^{\omega}}, \bm{W^{\omega}}, \bm{J^{\omega}}} \hspace{0.5cm} &  R_n [\Pi_n(\omega)],   \\
s.t. \hspace{0.5cm} &  (D_n, G_n, \bm{q}_n, u_n, W_n, J_n) \in \hat{\mathcal{K}}_n(\boldsymbol{x}_{-n}).
\end{align}
\end{subequations}

\begin{remark}
     Using the notion of Generalized Potential Games (GPG), as defined in \cite{facchinei}, we can show that (\ref{eq:reformulation_risk_trading_decentralized}) is a GPG. Indeed, the objective functions in (\ref{eq:reformulation_risk_trading_decentralized}) do not depend on the other players’ variables, so that the interaction of the players takes places only at the level of feasible sets. Thus, the potential function is given by the sum of the objective functions of all players, that is $P = \sum_{n\in\mathcal{N}} R_n[\Pi_n]$. This means that \eqref{eq:reformulation_risk_trading_decentralized} can be reformulated as an equivalent optimization problem, with $P$ in the objective \cite{facchinei}. Due to computational and communication limitations, it is not always possible to solve a large-scale optimization problem, and it is preferable to decompose the problem so that it can be solved by a distributed algorithmic approach. To that purpose, we can implement Regularized Gauss-Seidel best-response algorithm for the problem \eqref{eq:reformulation_risk_trading_decentralized} which is proven to converge for GPG in \cite{facchinei}. Recent research has focused on distributed algorithms for computing GNEs, so that coordination and information sharing within the agents are limited. Specific results exist for aggregative games \cite{paccagnan}. Recently, it has been proven that distributed learning approaches might also be relevant for games with strictly convex potential functions \cite{tatarenko}. Since none of these settings straightforwardly apply to our problem, we foresee it as future research direction for our work.        
\end{remark}

\section{Numerical Results} \label{numerical}

\subsection{Three node case}

In this section, we first present a toy network with only three nodes indexed by $\{0, 1, 2\}$. As in Proposition \ref{simplified VE coincidence for risk averse} we assume that $a_{mn} = a_{nm} = 0, \,\, \forall m,n \in \mathcal{N}$. The root node 0 has only conventional generation $\Delta G_0 = 0$, with cost $(a_0, b_0) = (4,30)$ and $(\underline{G}_n, \overline{G}_n) = (0,10)$ for $n \in \{0,1,2\}$. $\Delta G_n > 0$ for $n \in \{1, 2\}$. Each node is a consumer with $(\underline{D}, \overline{D}) = (0, 10)$ and generator (RES or conventional), therefore producing energy that can be consumed locally to meet demand $D_n$ and exported to the other nodes to meet the unsatisfied demand.

Regarding the preferences $(a_{nm})_{nm}$, nodes 1 and 2 both prefer to buy local and to RES-based generators. Node 0 is assumed to be indifferent between buying energy from node 1 or node 2. Capacities are also defined larger from the source node 0 $(\kappa_{0n} = 10)$ than between the prosumers nodes $(\kappa_{nm} = 5)$.

In Table \ref{table 1} we depict difference in values of the cost function in different frameworks of the market : risk-neutral (\textbf{RN}); risk averse homogeneous in agents' risk attitudes (\textbf{RA}): $(\chi_0,\chi_1,\chi_2) = (0.3,0.3,0.3)$; risk averse with one agent deviating (\textbf{RA with dev.}) in her risk attitude: $(\chi_0,\chi_1,\chi_2) = (0.3,0.9,0.3)$; risk-heterogeneous (\textbf{RH}) with $(\chi_0,\chi_1,\chi_2) = (0.3,0.4,0.5)$ and risk-heterogeneous with $(\chi_0,\chi_1,\chi_2) = (0.3,0.4,0.5)$ and risk-hedging contracts (\textbf{RH with contr.}).

\begin{table}

\caption{Agents' expected costs for different risk settings.}
\label{table 1}
\begin{center}
\begin{tabular}{|c|c|c|c|c|c|}
\hline
        \textbf{Node}     & \textbf{RN}     & \textbf{RA} & \textbf{RA  dev.} & \textbf{RH} & \textbf{RH contr.}    \\ \hline \hline
        \textbf{0}         & -6.1191        & -10.0359  & -10.0359 & -10.0359 & -14.0763      \\ 
        \textbf{1}         & -22.3625       & -13.3460  & -13.3311 & -13.3437 & -9.8183       \\ 
        \textbf{2}         & -7.2290        &  -7.1371    & -7.1445 & -7.1044 & -6.6550      \\ \hline
        \textbf{SC}        & -35.7107        & -30.5191   & -30.5116 & -30.4842 & -30.5497 \\ \hline
\end{tabular}
\end{center}
\end{table}

As it is depicted in Table \ref{table 1}, introduction of risk attitudes affects market in a negative way, decreasing the profit of the community. Changes in risk attitudes do not affect the profit of node 0 neither in risk-heterogeneous case, neither in the case of one node deviation. Deviation of node 1 from $\chi_1 = 0.3$ to $\chi_1 = 0.9$ in RA with dev. causes small decrease in her utility, but increased utility of node 2, leading to the small decrease of utility of the whole community. Risk-heterogeneity decreases the profit of node 2 with highest risk attitude $\chi_2 = 0.5$ and leads to a small decrease of community expected social cost.

Introduction of contracts leads to the biggest increase of the community's utility comparing to the simple RH or RA cases as it allows node 0 to sell ``surplus of utility'' to node 1 and to a smaller extent to node 2. As it is very costly for node 0 to produce energy, it is more profitable for it to hedge risk of other nodes, when the price $\alpha^{\omega} = \frac{p_{\omega}}{1-\min_{n}\chi_n}$ and $\gamma =\frac{p_{\omega}}{1-\max_{n}\chi_n} $. Both node 1 and node 2 benefit from buiyng contracts from node 0 as the price $\alpha$ is set to be profitable for both.

\begin{table}

\caption{Fincancial contracts trading.}
\label{table 2}
\begin{center}
\begin{tabular}{|c|c|c|c|c|c|}
\hline
\textbf{Node}          &  \multicolumn{3}{|c|}{W}       \\ \hline 
        \textbf{Scenarios}     & $p_0$ & $p_1$ & $p_2$           \\ \hline
        \textbf{0}            & -9.53446 &  -9.4188   & -9.44783    \\ 
        \textbf{1}           & 8.24404 & 8.35619    & 8.38462     \\ 
        \textbf{2}            & 1.29042 &  1.06261   & 1.06321     \\ \hline
        $\alpha^{\omega}$ & 0.47618       & 0.47618 &  0.47618         \\ \hline
        $\tau^{\omega}$    & 0.47618 & 0.0476    & 0.47618   \\ \hline
        $\phi^{\omega}$ & 0.0 & -0.42855 & 0.0 \\\hline
\end{tabular}
\end{center}
\end{table}

Note that $\tau^{\omega} = \frac{p_{\omega}}{1-\min_{n}\chi_n}$ for $\omega \in \{1,2\}$ and it corresponds to regime 2.b described in Proposition \ref{chracterization of the risk trading price}: indeed, for node 0 ($\chi_0 = \min_{n} \chi_n$) dual variables $\pi^{0}_0 = \pi^{2}_0$ equal $0$, which means that node 0 $\in \mathcal{M}$.

\newcommand{\hh}{\hspace{-2pt}} 
\tikzstyle{line} = [draw,thick=2, color=green!50, -latex',text=black,sloped ]

\tikzstyle{linec} = [draw,thick=2, color=red!70, -latex',text=black,sloped ]

\tikzstyle{lineb} = [draw,thick=2, color=blue!70, ->latex',text=black, sloped ]

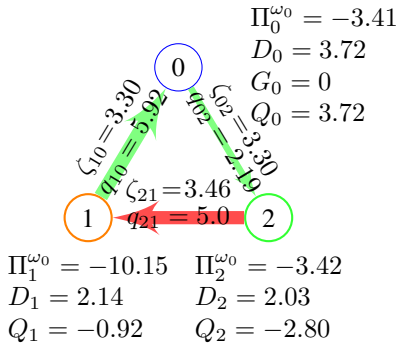
\begin{figure}[!htbp]
\begin{center}
\caption{Risk averse ($SC = -16.99$)}{
\begin{tikzpicture}[scale=1.0, >=latex]
\node [draw,circle,blue,fill=white,text=black] (n0) at (0,2) {0};

\node [draw,thick,circle, orange,text=black] (n1) at (-1.2,0) {1};

\node [draw,thick,circle, green!80,text=black] (n2) at (1.2,0) {2};
 \node [ right= 0.3cm of n0] (n0p)  {\begin{tabular}{l} 
 $\Pi^{\omega_0}_0 = -3.41$  \\
     $ D_0=3.72$ \\
  $ G_0=0$ \\
$ Q_0=3.72  $
\end{tabular} };
\node [below= (-0.cm) of n1] (n1p)  {\begin{tabular}{l} 
 $\Pi^{\omega_0}_1 = -10.15 $  \\
     $ D_1=2.14$ \\
  $ Q_1=-0.92  $
\end{tabular} };
\node [below = (-0.cm) of n2] (n2p)  {\begin{tabular}{l} 
 $\Pi^{\omega_0}_2 = -3.42 $  \\
     $ D_2=2.03$ \\
  $ Q_2=-2.80  $
\end{tabular} };
\path[line,line width=2.48pt](n0) -- node [] {$q_{02}=2.19$} node [above] { $\zeta_{02}\hh=\hh 3.30 $}  (n2) ;
\path[line,line width=5.38pt](n1) -- node [] {$q_{10}=5.92$} node [above] { $\zeta_{10}\hh=\hh 3.30 $}  (n0) ;
\path[linec,line width=5.0pt](n2) -- node [] {$q_{21}=5.0$} node [above] { $\zeta_{21}\hh=\hh 3.46 $}  (n1) ;
\end{tikzpicture}
}
\label{figure 1}
\end{center}
\end{figure}

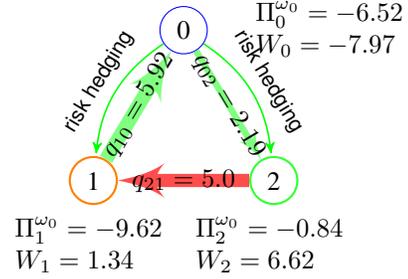
\begin{figure}[!htbp]
\begin{center}
\caption{RA with contracts ($SC = -16.99$)}{
\begin{tikzpicture}[scale=1.0, >=latex, 
            to/.style={->,>=stealth',shorten >=1pt,semithick,font=\sffamily\footnotesize},
  every node/.style={align=center}]

\node [draw,circle,blue,fill=white,text=black] (n0) at (0,2) {0};

\node [draw,thick,circle, orange,text=black] (n1) at (-1.2,0) {1};

\node [draw,thick,circle, green!80,text=black] (n2) at (1.2,0) {2};
 \node [ right= 0.3cm of n0] (n0p)  {\begin{tabular}{l} 
 $\Pi^{\omega_0}_0 = -6.52$  \\
     $ W_0= -7.97$ 
\end{tabular} };
\node [below= (-0.cm) of n1] (n1p)  {\begin{tabular}{l} 
 $\Pi^{\omega_0}_1 = -9.62 $  \\
     $ W_1= 1.34$ 
\end{tabular} };
\node [below = (-0.cm) of n2] (n2p)  {\begin{tabular}{l} 
 $\Pi^{\omega_0}_2 = -0.84 $  \\
     $ W_2=6.62$ 
\end{tabular} };
\path[line,line width=2.48pt](n0) -- node [] {$q_{02}=2.19$}  (n2) ;
\draw[to] [green](n0) to [bend right = 25] node [midway, above, black, sloped] {risk hedging} (n1);
\path[line,line width=5.38pt](n1) -- node [] {$q_{10}=5.92$} (n0) ;
\draw[to] [green] (n0) to [bend left = 25] node [midway, above, black, sloped] {risk hedging} (n2);
\path[linec,line width=5.0pt](n2) -- node [] {$q_{21}=5.0$}  (n1) ;
\end{tikzpicture}
}
\label{figure 2}
\end{center}
\end{figure}

In Table \ref{table 3}  we consider three node network in which now nodes 1 and 2 are prosumers with RES-based generators: $(\underline{G}_1, \overline{G}_1) = (\underline{G}_2, \overline{G}_2) = (0,0)$ for $n \in \{1,2\}$. 

We can observe, that now node 0 becomes more sensitive to the changes of node 1 and node 2 risk attitudes: both deviation of node 1 (RA with dev.: $(\chi_0,\chi_1,\chi_2) = (0.3,0.9,0.3)$) and risk heterogeneity ($(\chi_0,\chi_1,\chi_2) = (0.3,0.4,0.5)$) have higher impact comparing to Table \ref{table 1}.

As illustrated in the table, the introduction of contracts changes the behavior of prosumers as now again with price $\alpha^{\omega}$ set to be $\frac{p_{\omega}}{1-\min_{n}\chi_n}$ it is more profitable for node 0 to sell risk-hedging contracts. On the other side, nodes 1 and 2 tend to hedge their risk coming from RES-based generation for cheap price (from their perception of risk). Also, we may note, that introduction of contracts put \textbf{SC} of \textbf{RH} case to be equal to \textbf{RA case} where agents have equal risk attitude parameters $\chi_n$. 
Figure \ref{figure 1} illustrates the network for the setting (for one scenario $\omega = \omega_0$) described in Table \ref{table 3}: node 1 and node 2 have only RES-based generation. Note that the trade from node 1 to node 2 is done at full capacity. 

In Figure \ref{figure 2}, we show the same network (for one scenario $\omega = \omega_0$), but with contracts included in the agents' optimization problems. Note that we have a similar situation as the one depicted in Table \ref{table 2}: again, it is more profitable for node 1 to sell risk-hedging contracts than to produce energy due to the high conventional generation costs.

\begin{table}

\caption{Node 1 and 2 only RES-based production.}
\label{table 3}
\begin{center}
\begin{tabular}{|c|c|c|c|c|c|}
\hline
 \textbf{Node}     & \textbf{RN}     & \textbf{RA} & \textbf{RH} & \textbf{RA dev.} & \textbf{RH contr.}     \\ \hline \hline
        \textbf{0}         & -3.55       & -3.4207  & -3.4145 & -3.4201 & -6.5273      \\ 
        \textbf{1}         &  -10.18      & -10.1514 &  -10.1545 & -10.1522 & -9.6289       \\ 
        \textbf{2}         & -3.52       & -3.4258    & -3.4210 & -3.4255 & -0.8417     \\ \hline
        \textbf{SC}     & -17.25        & -16.9980   &-16.9901 & -16.9979  & -16.9980      \\ \hline
\end{tabular}
\end{center}
\end{table}

\subsection{14-bus network}
In this example, we consider the IEEE 14-bus network system as in \cite{le cadre}. Each pair of busses is able to trade with her neighboring busses, up to the capacity of the edge linking the pair of busses. We consider two cases: \textbf{(1)} $c_{nm} = c_{mn} = 1$ and \textbf{(2)} $c_{nm}$ 

are taken as in \cite{le cadre} - $b_{nm}$ are chosen uniformly in $[0,1]$ if $n \not= 1, m \not= 1$, and with assumption that agents have a preference for local trades so the price with the grid connection bus $c_{n1}$ chosen uniformly in $[1, 2]$. The grid connection bus has no preferences so that $c_{1n} = 1$ for
each $n$ neighboring bus 1.

\begin{figure}
\centering
\begin{minipage}{.45\linewidth}
  \includegraphics[width=\linewidth]{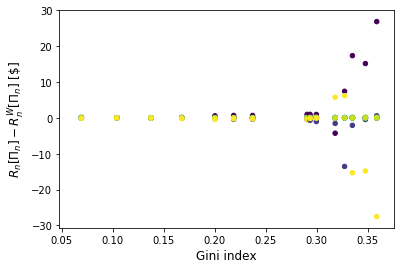}
  \caption{Uniform prices}
  \label{figure 3}
\end{minipage}
\hspace{.05\linewidth}
\begin{minipage}{.45\linewidth}
  \includegraphics[width=\linewidth]{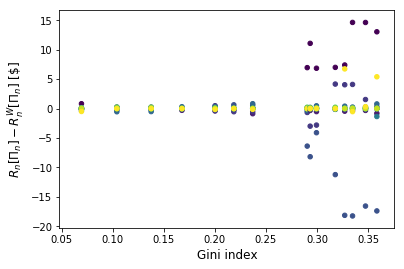}
  \caption{Heterogeneous prices}
  \label{figure 4}
\end{minipage}
\end{figure}
As depicted on the Figures \ref{figure 3} and \ref{figure 4}, increasing of heterogeneity in risk attitudes of the agents leads to more active trading of financial contracts. Heterogeneity in the trading costs puts financial restrictions on energy trading, that lead to activation of risk-hedging on the lower values of Gini index. As in the case of three nodes, node 0 actively sells contracts, as it is more profitable for her to sell contracts than to produce energy.

\section{CONCLUSION}

We formulate two market design (centralized used as benchmark, and peer-to-peer) under risk neutral and risk averse settings. We analyse the market equilibria for these two settings and characterize the impact of strategic behaviors of the agents and risk attitude heterogeneity on the market equilibria, relying on Generalized Nash Equilibrium and Variational Equilibrium as solution concepts. We consider the inclusion of risk hedging mechanisms in the incomplete energy trading market, in the form of financial contracts, and provide closed form characterizations of the risk-adjusted probabilities under different market regimes. 

As next step, we will focus on the development of distributed learning algorithms for risk trading mechanism, that minimize the information exchanges among the prosumers. Some results already exist for aggregative games or under strict convexity of the potential function, which do not hold here. Another research direction would be to consider the strategic behavior of an aggregator located at the root node.

\end{document}